\documentclass[preprint,authoryear]{elsarticle}
\usepackage[cmex10]{amsmath}
\usepackage{amsfonts}
\usepackage{amssymb}
\usepackage{amsopn}
\usepackage{graphicx}
\usepackage{subfigure}
\usepackage{natbib}
\usepackage{amsthm}

\graphicspath{{fig/}}

\newcommand{\sups}[1]{\ensuremath{^{\textrm{#1}}}}

\theoremstyle{plain} \newtheorem{splcases}{Theorem}
\theoremstyle{plain} \newtheorem{errbound}[splcases]{Theorem}

\begin{document}

\title{Approximation of Average Run Length of Moving Sum Algorithms Using Multivariate Probabilities}
\author{Swarnendu Kar\corref{cor1}}
\ead{swkar@syr.edu}
\author{Kishan G. Mehrotra}
\ead{mehrotra@syr.edu}
\author{Pramod K. Varshney}
\ead{varshney@syr.edu}

\address{Department of Electrical Engineering and Computer Science\\ Syracuse University, Syracuse, NY, 13244}

\cortext[cor1]{Corresponding author. Tel: (315) 751-1370 Fax: (315) 443-4745}

\begin{abstract}
Among the various procedures used to detect potential changes in a stochastic process the moving sum algorithms are very popular due to their intuitive appeal and good statistical performance. One of the important design parameters of a change detection algorithm is the expected interval between false positives, also known as the average run length (ARL). Computation of the ARL usually involves numerical procedures but in some cases it can be approximated using a series involving multivariate probabilities. In this paper, we present an analysis of this series approach by providing sufficient conditions for convergence and derive an error bound. Using simulation studies, we show that the series approach is applicable to moving average and filtered derivative algorithms. For moving average algorithms, we compare our results with previously known bounds. We use two special cases to illustrate our observations.
\end{abstract}


\maketitle

\section{Introduction}
The problem of detecting a change in the mean value of a process, when both the change point and change magnitude are unknown, is of great importance in various disciplines such as econometrics, engineering, quality control, technical analysis of financial data, and edge detection in image processing.
The optimal scheme, which involves Maximum-Likelihood estimation of both the change point and the change magnitude, is computationally prohibitive. Hence various simple but suboptimal methods like ordinary moving average (MA),  exponentially weighted moving average (EWMA), and filtered derivative (FD) are used in practice. For example, MA is used in technical analysis of financial data, like stock prices, returns or trading volumes \citep*{Mur99}.

To apply the MA scheme a finite, (say $k$),  immediate past samples are added with equal weights $1/k$, while in EWMA,  the past samples are combined with exponentially decreasing weights. A generalization of the MA scheme is where the past $k$ samples are combined with arbitrary positive weights (see  \citep*{Lai74,Bohm90}). In other applications, e.g., filtered derivative in edge detection \citep*{Bas81}, the past $k$ samples are combined with both positive and negative weights. In this paper, such generalizations are referred to as moving sums (MOSUM). In this paper, we study the most commonly used MOSUM algorithms - namely the MA and FD schemes.

Change detection schemes are assessed on the basis of the statistical distribution of the run length, i.e., the number of test samples taken before a false positive is detected. For most practical purposes, the distribution function of run length is adequately summarized by its expected (or average) value, also known as the average run length (ARL). Unfortunately, for most practical schemes closed form expression for ARL is difficult to obtain. For the MA scheme, the bulk of the research so far has been dedicated to either tabulating numerical results through Monte Carlo simulations \citep*{SQC04} or deriving bounds using multivariate probability distribution functions (MPDFs) \citep*{Bohm90}. Very little work has been done to date regarding the ARL of FD scheme.

The ARL can be written as the sum of an infinite number of MPDFs of increasingly higher dimensions. But MPDFs, in general, can be computed only numerically and the computational intensity increases with the dimension of the multivariate vector. While addressing the problem of approximating the ARL for EWMA algorithms, it was observed by \cite{Rob78} that the ratio of the successive MPDFs converge as the dimension increases. This fact was used to propose a series based approach of approximating the ARL. For the MOSUM algorithm with positive weights and by using only $k$\sups{th} order MPDFs, upper and lower bounds were proposed by \cite{Lai74} and improved by \cite{Bohm90}. This was based on the idea that an MPDF of dimension larger than $k$ can be bound (both above and below) by products of lower order MPDFs.  Consistent with the name of the authors, we refer to those results as LBH bounds.

In this paper, we analyze the series approach of \cite{Rob78} by laying out sufficient conditions for convergence and also provide an error bound. With simulation studies, we demonstrate the versatility of this approach by showing that the conditional survival probabilities also converge for MA and FD algorithms. For MA algorithms, we compare the actual convergence vis-a-vis the LBH bounds. Through simulation studies for both MA and FD schemes, we demonstrate that a satisfactory approximation of the ARL can be obtained by using only $\lceil k/2 \rceil$\sups{th} order MPDFs. Compared to $k$\sups{th} order MPDFs as required by the LBH bounds, this provides a significant saving in computation for the MA scheme.

The rest of this paper is organized as follows. We introduce some notations in Section \ref{sec:notations}. In Section \ref{sec:mainresults}, we summarize the main results of this paper. We examine the convergence of conditional survival probabilities for MOSUM algorithms in Section \ref{sec:conv:rn}. In Section \ref{sec:comp:approx}, we compare the convergence of ARL for MA algorithms with the LBH bounds. We also demonstrate that MPDFs of $\lceil k/2 \rceil$\sups{th} order provides a reasonable approximation of the ARL. Concluding remarks are provided in Section \ref{sec:conc}.

\section{Notations} \label{sec:notations}
Let $X_1,X_2,\ldots X_m, \ldots $ be a sequence of observations obtained from a discrete time random process. We assume that  $X_i$'s are independently distributed random variables with variance $\sigma^2$. It is assumed that the mean of $X_i$'s possibly changes from $\mu$ to $\theta$ at some point, here
$\sigma^2$,  $\mu$ and $\theta > \mu$ are unknown parameters and possible time of change is also unknown.  For detecting a change in the mean value of $X_i$'s from we need to formulate an appropriate linear test statistic, say $Y_m = \sum_{i=m-k+1}^m c_{m-i}X_i$, where $c_i$ are constants and compare it against some threshold. \cite{Rob59} has considered both the MA and EWMA schemes for appropriately  chosen weights.
A generalization of MA was considered by \cite{Bohm90}, where $c_i \geq 0$ for all $i$'s.  There are other applications where all the weights need not be positive. For example, in the context of edge detection, \cite{Bas81} uses the following test statistic
\begin{align*}
Y_m = \sum_{i=m-k+1}^{m-k/2} X_i - \sum_{i=m-k/2+1}^m X_i
\end{align*}
where $k$ is assumed to be an even number. This is also known as the filtered derivative (FD) algorithm, since we take the difference of averaged (filtered) block of samples.

To test whether the process mean is  $\mu$ or has shifted to $\theta$, the test statistic $Y_m$ is monitored  for successive values of $m$  and  compared against  an upper  threshold, say $h$.  The threshold is sometimes also specified as multiples of the standard deviation in excess of the mean of the test statistic, i.e. by the quantity $\delta$ defined by
\begin{align*}
\delta &= \frac{h-\mathbf{E}(Y_m)}{\sqrt{\text{Var}(Y_m)}}, \\
       &= \frac{h-\mu\sum c_i}{\sigma \sqrt{\sum c_i^2}}.
\end{align*}
In this paper, $\text{MOSUM}\left([c_0,c_1,\ldots,c_{k-1}],\delta \right)$
denotes  a moving sum algorithm with weights $c_0,c_1,\ldots,c_{k-1}$ and threshold $\delta$.

The time elapsed before $Y_m$ exceeds the thresholds for the first time is also known as the run length (RL) or stopping time.
\begin{align*}
RL=\max\{m:Y_m <h \}
\end{align*}
In this paper, we are interested in the average run length (ARL), namely the expectation of \emph{RL}, i.e., in
\begin{align*}
L=\mathbf{E}(RL)
\end{align*}

To show that $L$  can be represented as an infinite sum of probability distributions we define:
\begin{itemize}
\item $p_n$, the probability that RL$=n$, i.e.,  the test passes $n-1$ consecutive times but fails at the $n$\sups{th} instant;
\begin{align*}
 p_n = P(Y_{k+1},Y_{k+2},\ldots,Y_{k+n-1}<h,Y_{k+n}>h).
\end{align*}
\item $q_n$, the survival probability is the probability that the test passes $n$ consecutive times
\begin{align*}
 q_n = P(Y_{k+1},Y_{k+2},\ldots,Y_{k+n}<h)
\end{align*}
\item $r_n$, the conditional survival probability that the test survives at a particular instant given that it has already survived the past $n-1$ times,
\begin{align*}
 r_n = P(Y_{k+n}<h|Y_{k+1},Y_{k+2},\ldots,Y_{k+n-1}<h)
\end{align*}
\end{itemize}
It follows from these definitions that $p_n=q_{n-1}-q_n$ and $r_n=q_n/q_{n-1}$. Since the evaluations start at index $k$, the ARL function can be represented as
\begin{align}
L &=k-1+\sum_{n=1}^{\infty}np_n \nonumber \\
  &=k+\sum_{n=1}^{\infty}q_n  \label{Linfseries}
\end{align}
In this paper, we use \eqref{Linfseries} to either derive closed form expressions or provide approximate results.

\section{Main results}  \label{sec:mainresults}

In Theorem \ref{thm:splcases}, we obtain closed form expressions for $L$ for two special cases.
\begin{splcases}[Two special cases] \label{thm:splcases}
Let $X_1, X_2, \ldots X_n, \ldots$ be zero-mean i.i.d. random variables with symmetric pdf. Then the following results apply,
\begin{enumerate}
\item For MOSUM$\left([-1,1],0\right)$, $L=\exp(1)\approx 2.7183$.
\item For MOSUM$\left([1,1],0\right)$, $L=\sec(1)+\tan(1)\approx 3.4082$.
\end{enumerate}
\end{splcases}

\begin{proof}

\begin{enumerate}
\item For $c_0=1$ and $c_1=-1$, the $q_n$ is given by
\begin{align}
q_n&= P\left(X_1-X_2<0,X_2-X_3<0,\ldots,X_n-X_{n+1}<0\right), \nonumber \\
   &= P\left(X_1<X_2<X_3\cdots <X_{n+1}\right). \label{pf:thm2:lbl1}
\end{align}
We recall that $X_1,X_2,\ldots,X_{n+1}$ are i.i.d. random variables. If we draw $n+1$ independent samples from the same distribution and order them, they can result in any one of the $(n+1)!$ possible orderings with equal probability. Since \eqref{pf:thm2:lbl1} denotes only one such ordering, we conclude that
\begin{align*}
q_n&= \frac{1}{(n+1)!}.
\end{align*}
This result is well known for the Gaussian random variables \citep*{Bar72}. Using \eqref{Linfseries}, we obtain
\begin{align*}
L&= 2+\sum_{n=1}^{\infty} \frac{1}{(n+1)!} =  \exp(1)
\end{align*}

\item Next we consider $c_0=1,c_1=1$. The MPDF $q_n$ was derived in the context of Gaussian random variables in \citep*{Mor83}, but a careful analysis of the proof reveals that  the argument is valid for any symmetric distribution.  Thus, quoting \citep*{Mor83}, $q_n$ is given by the coefficient of $z^{n+1}$ in the power series expansion of $\sec(z)+\tan(z)$ around the point $z=0$. Since,
\begin{align*}
\sec(z)+\tan(z)=1+z+\sum_{n=1}^{\infty}q_n z^{n+1}
\end{align*}
Evaluating the above series at $z=1$ and using \eqref{Linfseries}, we obtain
\begin{align*}
L&= \sec(1)+\tan(1).
\end{align*}
\end{enumerate}
\end{proof}

These two special cases  may not be of any practical use, nevertheless, we use these cases as illustrative examples in our discussions in Section \ref{sec:conv:rn}.

\subsection{Previous work on approximating ARL using MPDFs}
We observe that $q_n$ is an MPDF of $n$-dimensions  involving correlated variables $Y_{k+1},Y_{k+2},\ldots,Y_{k+n}$. In general, closed form expressions for $q_n$ are not available; the values are computed numerically. The intensity of these computations increases with $n$. As a result, the summation in the form of \eqref{Linfseries} can seldom be used to compute the ARL. Various methods for approximating the ARL have been proposed by the researchers; we briefly describe two such approaches due to \cite{Lai74} and \cite{Rob78}  below.

For EWMA, and using simulations, \cite{Rob78} observed  that as $n$ increases the conditional survival probabilities $\{r_n\}$ appear to converge. If we assume that $r_i\approx r_n, \forall i>n$, then the future survival probabilities can be approximated as $q_i \approx q_{i-1}r_n, \forall i> n$. By using  \eqref{Linfseries}, we obtain the $n$\sups{th} order approximation of $L$ as:
\begin{align}
L_n &= k+\sum_{i=1}^{n-1}q_i+ q_n(1+r_n+r_n^2+\ldots) \nonumber \\
    &= k+\sum_{i=1}^{n-1}q_i+ \frac{q_n}{1-r_n}. \label{Lnseries}
\end{align}
Thus, in \eqref{Lnseries}, the ARL  is approximated by using only a few lower order MPDFs. Another significant result due to \cite{Lai74} and improved by \cite{Bohm90} provides an upper and a lower bound on ARL for MOSUM algorithms with positive weights, as follows:
\begin{align}
1+\frac{q_k}{p_k}\leq L \leq k+\frac{q_k}{p_k}. \label{mabound}
\end{align}
We denote the lower and upper bounds (together we call them LBH bounds) in \eqref{mabound} by $L_l$ and $L_u$ respectively. These bounds are significant because they are asymptotically the same, i.e.,  when ($L$) is large compared to the span $k$, the upper and lower bounds are almost equal. The ARL can then be approximated as $L\approx q_k/p_k$, which requires the computation of MPDFs of order $k$.

Very little work can be found in the literature, if any, regarding the approximation of ARL for the FD scheme.

\subsection{Convergence analysis of $\{r_n\}$ and $\{L_n\}$}
In Theorem \ref{thm:errbound}, we present an analysis of the series in \eqref{Lnseries}. In particular, we derive an error bound that relates the convergence of $L_n$ to that of $r_n$.

\begin{errbound}[Convergence and error bound] \label{thm:errbound}
Assume $L<\infty$, $r_n\rightarrow r$ for some $r\in [0,1)$. Then
\begin{enumerate}
\item  $L_n\rightarrow L$.
\item  Choose $\epsilon$ such that $0<\epsilon<1-r$. If $|r_n-r|<\epsilon$ \mbox{ for all } $n \geq m$, then
\begin{align}
\left | \frac{L_n-L}{L} \right| &<\frac{2\epsilon}{(1-r-\epsilon)^2} \mbox{ for all } n \geq m  \label{errbd:gen}
\end{align}
and if, in particular, convergence $\{r_n\}_{m+1}^{\infty}$ is monotonic, then
\begin{align}
\left | L_n-L \right| <\frac{\epsilon}{(1-r-\epsilon)^2} \mbox{ for all } n \geq m \label{errbd:mon}
\end{align}
\end{enumerate}
\end{errbound}

\begin{proof}

\begin{enumerate}
\item If  $L<\infty$, then from equation \eqref{Linfseries} it follows that  $q_n\rightarrow 0$. Taking limits on the right hand side of \eqref{Lnseries}, we complete the first result. That is:
\begin{align*}
\lim_{n\rightarrow\infty}L_n = L+\frac{\lim_{n\rightarrow\infty}q_n}{1-r} = L.
\end{align*}
\item
Choose $\epsilon >0$. Since $\{r_n\} \rightarrow r$, there is an index $m$ such that
\begin{align}
|r_i-r|<\epsilon, \qquad \forall i \geq m. \label{convcond}
\end{align}
To obtain the desired bounds on $L$ we first note that $L = \lim_{n\rightarrow \infty} L_n$ and write this limiting value as a telescopic sum to obtain:
\begin{align}
L & \equiv  L_m  + ( L_{m+1}-L_m)+(L_{m+2}-L_{m+1})+\cdots \nonumber \\
 &= L_m+\sum_{i=m+1}^{\infty} \left(L_i-L_{i-1}\right) \label{pf:thm1:lbl1}
\end{align}
For $i\geq 2$, we use \eqref{Lnseries} to obtain
\begin{align}
L_i-L_{i-1} &= q_{i-1}+ \frac{q_i}{1-r_i}-\frac{q_{i-1}}{1-r_{i-1}} \nonumber \\
            &= \frac{q_{i-1}}{1-r_i}-\frac{q_{i-1}}{1-r_{i-1}}      \nonumber \\
            &= q_{i-1}\frac{r_i-r_{i-1}}{(1-r_i)(1-r_{i-1})}        \label{pf:thm1:lbl2}
\end{align}
Using \eqref{pf:thm1:lbl2} in \eqref{pf:thm1:lbl1}, we can bound the approximation error as follows
\begin{align}
\left|L_m-L\right| &=\left| \sum_{i=m+1}^{\infty} q_{i-1}\frac{r_i-r_{i-1}}{(1-r_i)(1-r_{i-1})} \right|       \nonumber \\
 &\stackrel{\text{(a)}}{\leq}\sum_{i=m+1}^{\infty} q_{i-1}\frac{\left|r_i-r_{i-1}\right|}{(1-r_i)(1-r_{i-1})} \nonumber \\
 &\stackrel{\text{(b)}}{\leq} \frac{1}{(1-r-\epsilon )^2}\sum_{i=m+1}^{\infty} q_{i-1}\left|r_i-r_{i-1} \right|   \label{pf:thm1:lbl3}
\end{align}
where (a) follows since $q_i,1-r_i>0,\quad\forall i$, and (b) is due to our initial assumption that $1-r-\epsilon>0$ and from \eqref{convcond}, we can obtain that $1-r-\epsilon<1-r_i,\quad\forall i\geq m$.

From \eqref{pf:thm1:lbl3}, we can proceed to obtain \eqref{errbd:gen} as follows
\begin{align*}
\left| \frac{L_m-L}{L} \right|
&\stackrel{\text{(a)}}{<} \frac{1}{L}\frac{2\epsilon}{(1-r-\epsilon)^2} \sum_{i=m+1}^{\infty} q_{i-1}  \\
&\stackrel{\text{(b)}}{<} \frac{2\epsilon}{(1-r-\epsilon)^2}
\end{align*}
where (a) follows from the fact that for $i\geq m+1$, $|r_i-r_{i-1}|\leq |r_i-r|+|r_{i-1}-r|< 2\epsilon$ and (b) follows from the infinite sum in \eqref{Linfseries}.

From \eqref{pf:thm1:lbl3}, assuming $\{r_n\}_{m+1}^{\infty}$ to be monotonic, we can proceed to obtain \eqref{errbd:mon} as follows
\begin{align*}
\left|L_m-L\right|
 &\stackrel{\text{(a)}}{<} \frac{1}{(1-r-\epsilon )^2}\sum_{i=m+1}^{\infty} \left|r_i-r_{i-1} \right|         \\
 &\stackrel{\text{(b)}}{=} \frac{1}{(1-r-\epsilon )^2}\left|\sum_{i=m+1}^{\infty} (r_i-r_{i-1}) \right|       \\
 &\stackrel{\text{(c)}}{=} \frac{1}{(1-r-\epsilon )^2}\left|r-r_m \right|                                     \\
 &\stackrel{\text{(d)}}{<} \frac{\epsilon}{(1-r-\epsilon )^2}
\end{align*}
(a) is due to fact that $q_i<1$ $\forall i$, (b) follows from the monotonicity of $\{r_n\}_{m+1}^{\infty}$, (c) follows from an argument exactly similar to the one used to derive \eqref{pf:thm1:lbl1} and (d) is due to assumption \eqref{convcond}.
\end{enumerate}
\end{proof}

Applicability of Theorem \ref{thm:errbound} relies on the convergence of $\{r_n\}$. We discuss the convergence of $\{r_n\}$ for some moving sum algorithms in the following section.

\section{Convergence of $\{r_n\}$ for Moving Sums}   \label{sec:conv:rn}
We use simulation studies to examine the convergence properties of $r_n$ for MA and FD algorithms. In all our simulations, it was assumed that the error variables $X_i$ are normally distributed. The MPDFs were calculated as multivariate normal CDFs using the technique in \citep*{Genz92}. We have used the software available from the author's website. We evaluate $r_n$ for various values of the span $k$. We fix the standardized thresholds at $\delta = 0$ and $\delta = 2$. The graphs are displayed in Figure \ref{fig:conv:rn}.

\begin{figure*}[p!]
     \centering
     \subfigure[Moving average algorithm with low threshold ($\delta=0$)]{
          \label{fig:conv:rn:ma:d0}
          \includegraphics[width=.45\textwidth]{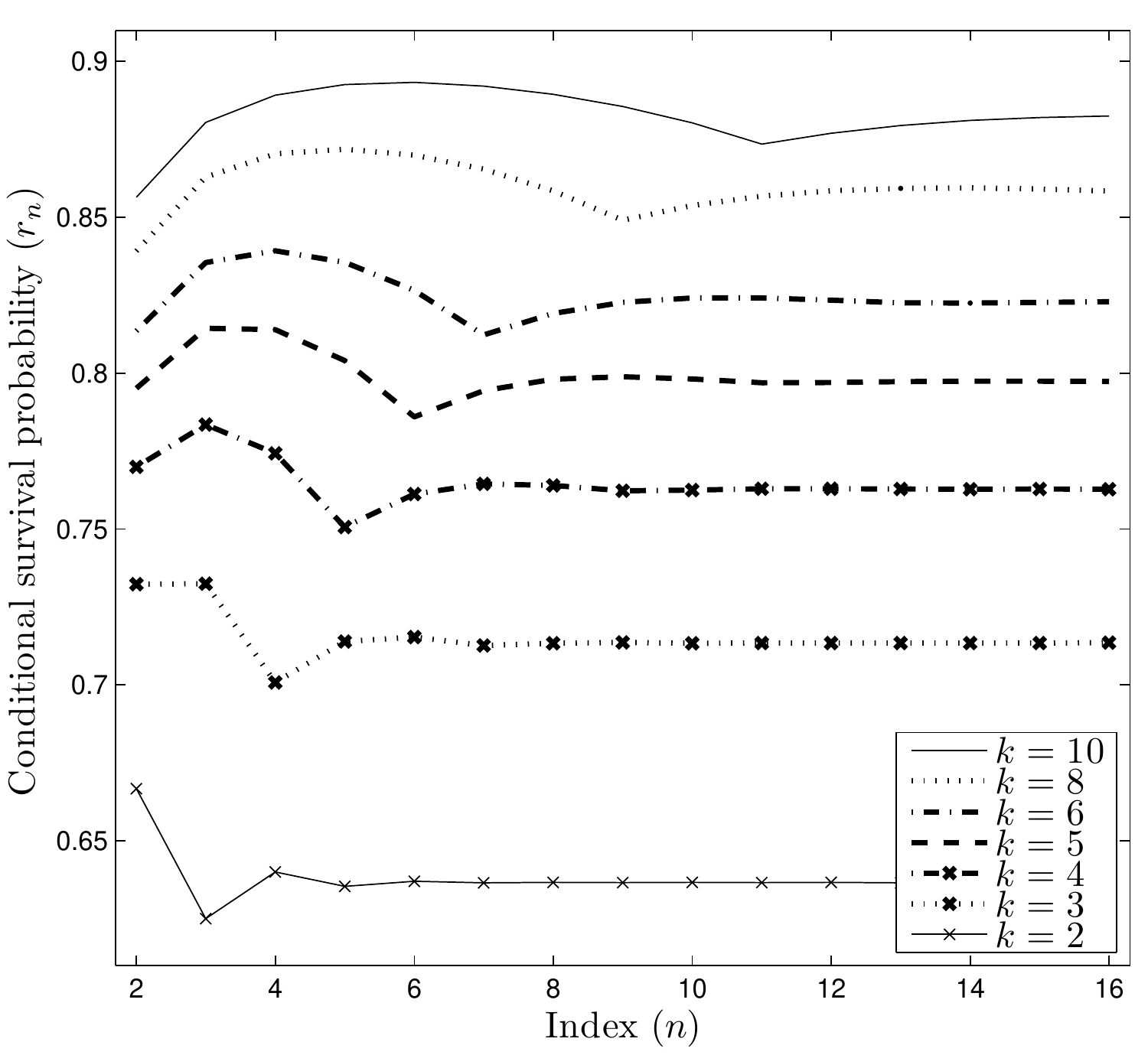}}
     \hspace{.3in}
     \subfigure[Moving average algorithm with high threshold ($\delta=2$)]{
          \label{fig:conv:rn:ma:d2}
          \includegraphics[width=.45\textwidth]{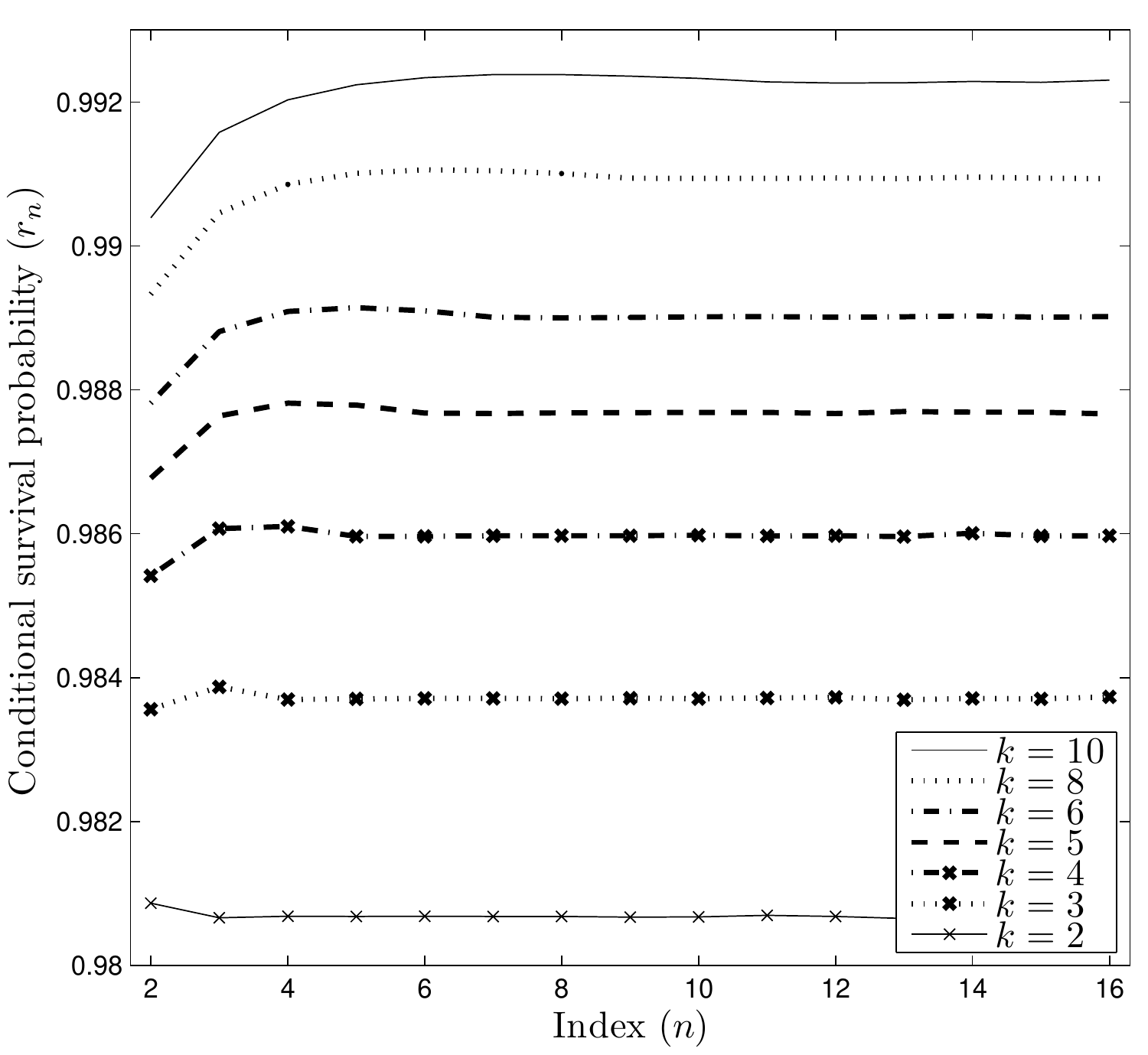}}\\
     \vspace{.3in}
     \subfigure[Filtered derivative algorithm with low threshold ($\delta=0$)]{
           \label{fig:conv:rn:fd:d0}
           \includegraphics[width=.45\textwidth]{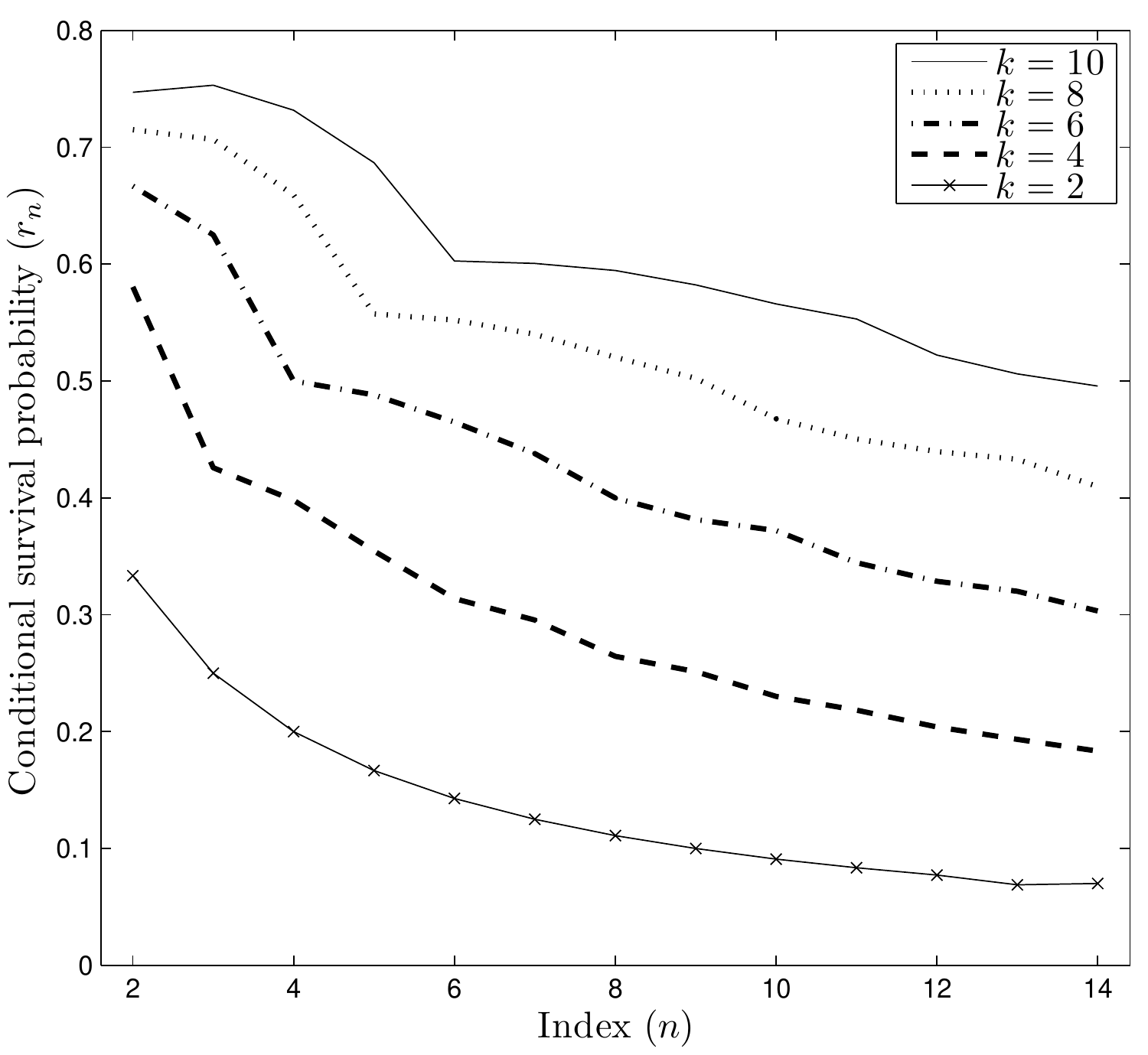}}
     \hspace{.3in}
     \subfigure[Filtered derivative algorithm with high threshold ($\delta=2$)]{
           \label{fig:conv:rn:fd:d2}
          \includegraphics[width=.45\textwidth]{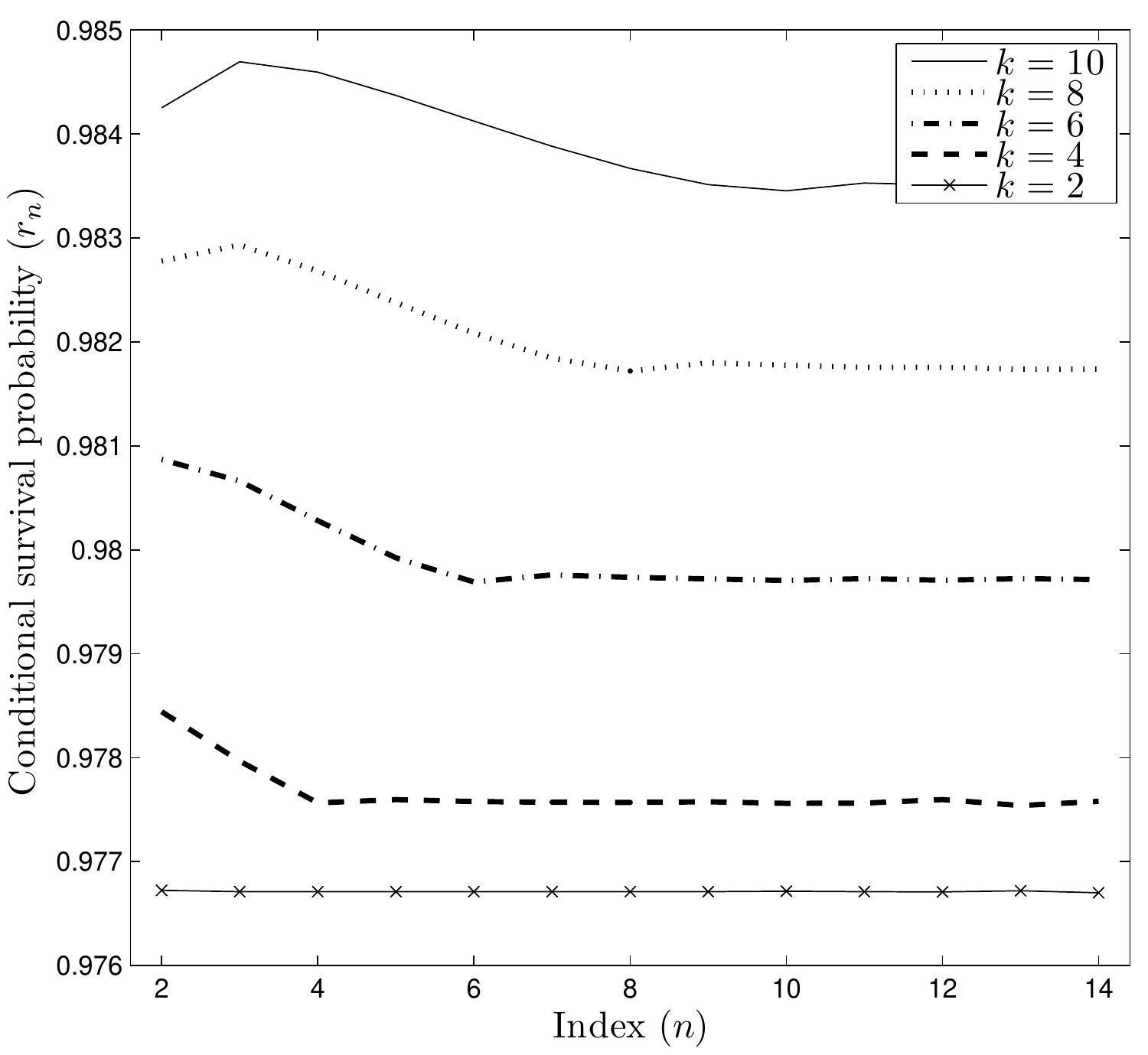}}
     \caption{Convergence of conditional survival probability ($r_n$) for one-sided moving average and filtered derivative algorithms considering low ($\delta=0$) and high ($\delta=2$) thresholds.}
      \label{fig:conv:rn}
\end{figure*}

In Figure \ref{fig:conv:rn:ma:d0} we observe that $r_n$ achieves its limiting value very quickly after some initial zig-zag pattern. This limiting value, $r$, depends on $k$ and it is observed that the length of zig-zagedness depends on the value of $k$. This behavior is more pronounced for $\delta = 0$.
For the case of filtered derivatives $r_n$ (see Figures \ref{fig:conv:rn:fd:d0} and \ref{fig:conv:rn:fd:d2}) monotonically decrease after initial minor deviations. Consequently, for the filtered derivative model the tighter bound, given by equation \eqref{errbd:mon}, can be applied. The monotonic convergence ($r_n\rightarrow 0$) follows for the special case of MOSUM$([-1,1],0)$ from Theorem \ref{thm:splcases} where  we had observed that  the survival probability $q_n = \frac{1}{(n+1)!}$ and therefore the conditional survival probability is $r_n = \frac{1}{n+1}$ which converges to zero.

The zig-zag property, which is most pronounced in Figure \ref{fig:conv:rn:ma:d0}, is easy to explain for the special case of MOSUM$([1,1],0)$. From Theorem \ref{thm:splcases} we know that  $q_n$ is the coefficient of $z^n$  in the power series expansion of $\sec(z)+\tan(z)$ around $z=0$.  Thus, we obtain the survival probabilities as:
\begin{align*}
\{q_n\} &= \left\{\frac{1}{2},\frac{1}{3},\frac{5}{24},\frac{2}{15},\frac{61}{720},\cdots \right\}.
\end{align*}
The conditional survival probabilities can therefore be computed as
$\{r_n\} = \left\{\frac{2}{3},\frac{5}{8},\frac{16}{25},\frac{61}{96},\cdots \right\} = \{0.667, 0.625, 0.640, 0.635, \cdots \}$. It is a known result \citep*{seqa11} that $r_n$ converges to $2/\pi=0.6366$. Also $\{r_n\}$ decreases and increases alternately, giving the series a zig-zag appearance of period $2$. For moving averages of span $k$, we observe that the zig-zag pattern persists with period $k$.

Based on the observations in Figure \ref{fig:conv:rn}, several conjectures can be made about the convergence of the conditional survival probabilities $\{r_n\}$ in MOSUM algorithms.
\begin{itemize}
\item The sequence $\{r_n\}$ converges to some $r$ as $n$ incresaes,
\item the limit point $r$ is an increasing function of both the $k$ and the threshold $\delta$,
\item the convergence to $r$ is faster for higher thresholds.
\end{itemize}
It was shown in Theorem \ref{thm:errbound} that convergence of $\{r_n\}$ is a sufficient condition for the applicability of the series approach of approximating the ARL. In the next section, we apply the series approach to MA and FD algorithms.

\section{Comparison with LBH bounds and approximation of ARL}   \label{sec:comp:approx}
First we consider moving average algorithms and compare the convergence of the series approach \eqref{Linfseries} with the LBH bounds given by \eqref{mabound}. Since the filtered derivative algorithms have negative weights, the LBH bounds are not applicable.

A representative example,  demonstrating the convergence of $\{L_n\}$, is shown is Figure \ref{fig:arlcomp:q95w8} for $k=8$. We considered the threshold $h$ such that that $q_1=P(Y_1<h) = 0.95$. We have also plotted the ARL obtained via Monte-Carlo simulations.  We observe that the series enters the region bounded by $L_l$ and $L_u$ fairly rapidly. Though the LBH bounds could only be calculated using MPDFs of $8$\sups{th} order, reasonable approximation can be obtained using lower order MPDFs. In Figure \ref{fig:arlcomp:k8}, we demonstrate this fact for different thresholds. Since the ARL varies with threshold, we compared them with their respective LBH bounds. We plot the ARL in excess of the corresponding lower LBH bound, i.e., we plot the quantity $L_u-L_l$. We note that $L_u-L_l=k-1$. The [$0,k-1$]-lines (dotted) represent $L_l$ and $L_u$ respectively.

\begin{figure}
\begin{center}
    \includegraphics[width=9cm]{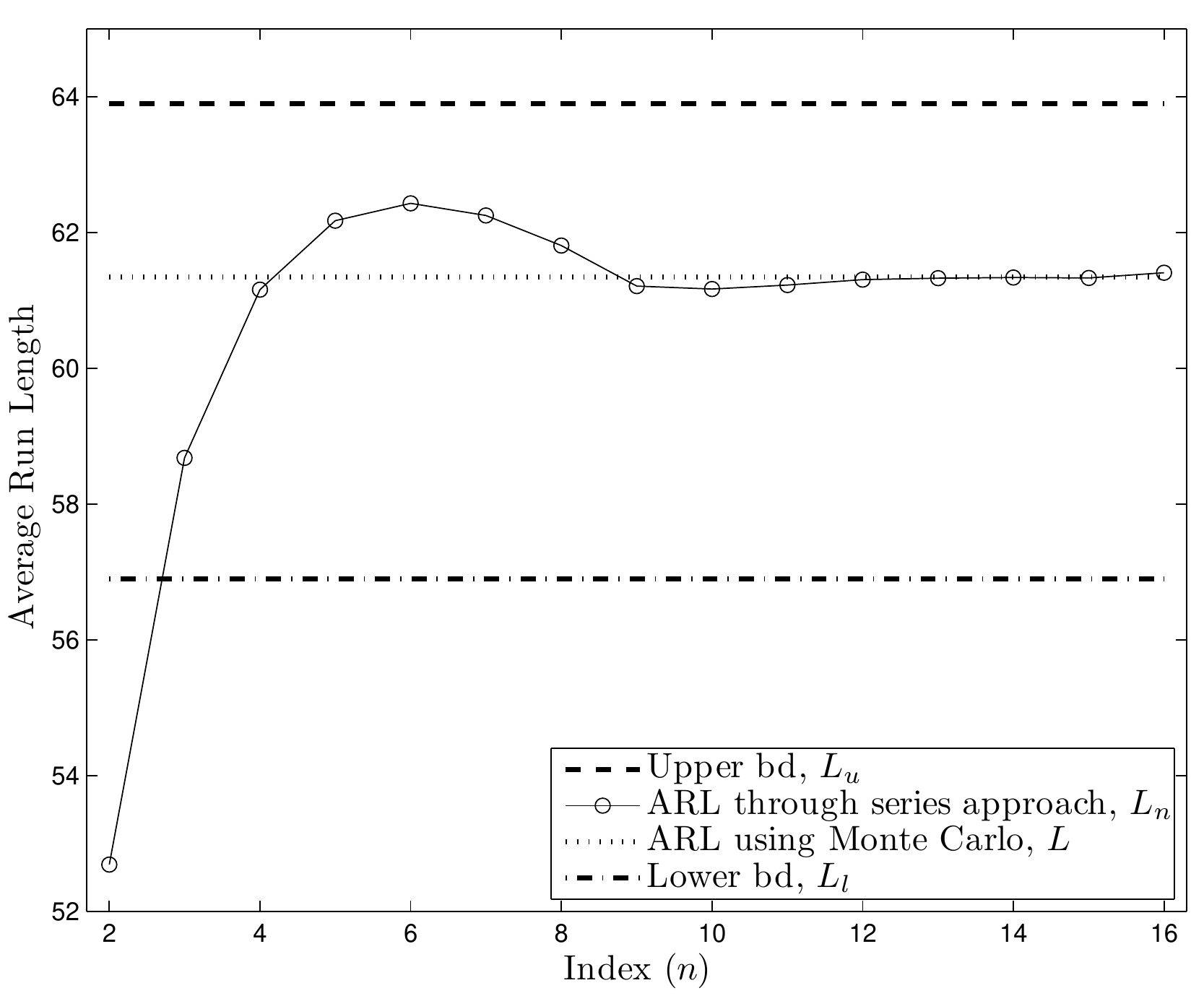}
  \caption{Convergence of $\{L_n\}$ and comparison with $L_l,L_u$ for span $k=8$ and threshold given by $q_1=0.95$}
  \label{fig:arlcomp:q95w8}
  \end{center}
\end{figure}

\begin{figure}
\begin{center}
    \includegraphics[width=9cm]{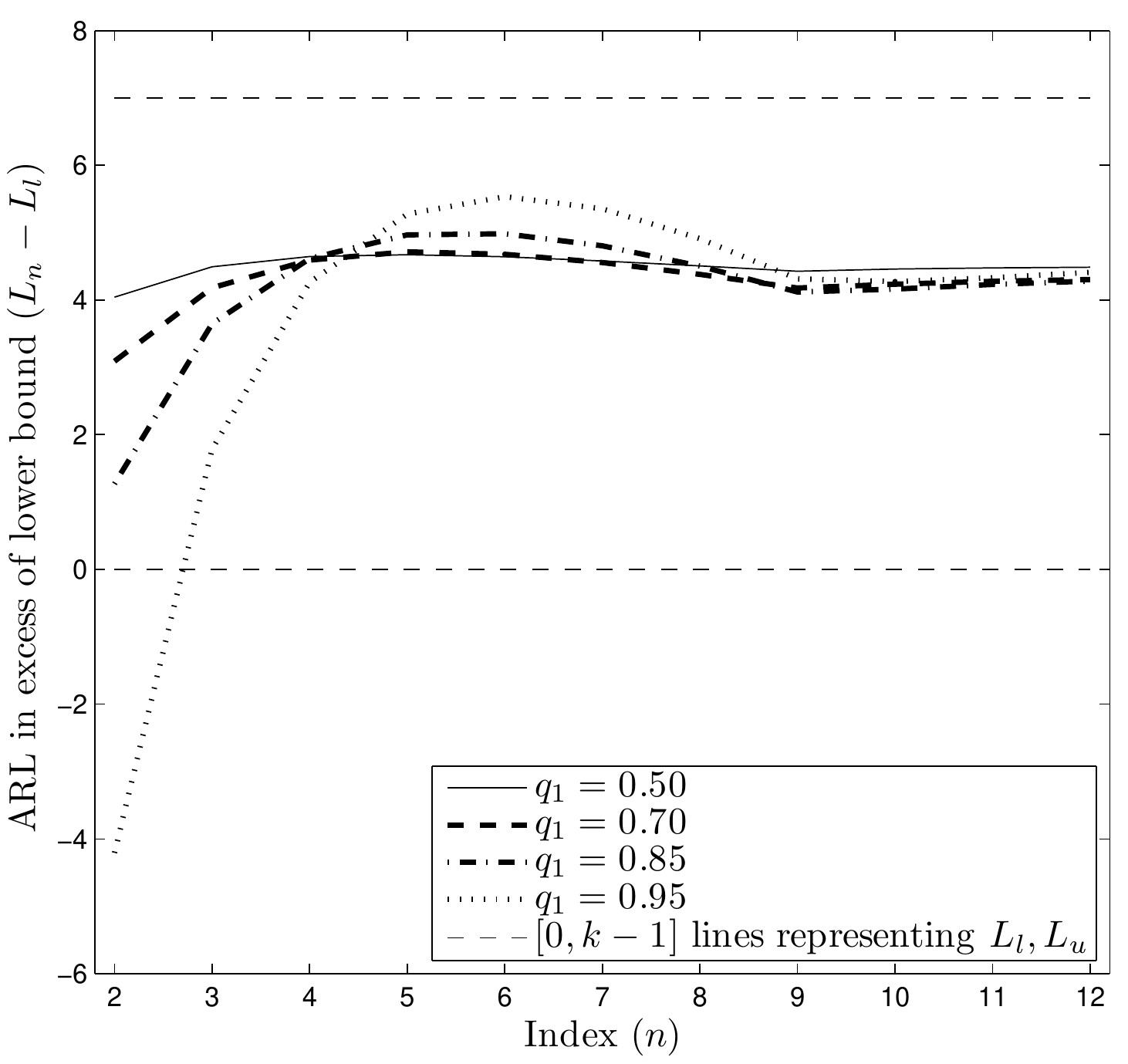}
  \caption{Convergence of $\{L_n\}$ and comparison with $L_l,L_u$ for span $k=8$ and various thresholds}
  \label{fig:arlcomp:k8}
  \end{center}
\end{figure}

We tabulate the ARL for moving average and filtered derivative algorithms in Table \ref{tbl:mafd:s1} and compare them with approximate values obtained using the series approach in \eqref{Lnseries}. We show the approximation $L_{\lceil k/2 \rceil}$, which uses only $\lceil k/2 \rceil$\sups{th} order MPDFs. For moving average algorithm with span $k\leq 10$, the reference values of $L$ were obtained from \cite{SQC04}. All other reference values were obtained using Monte-carlo simulations. We have considered only one-sided tests and the ARLs are tabulated in Table \ref{tbl:mafd:s1}. We conclude that $L_{\lceil k/2 \rceil}$ provides a reasonable approximation of $L$.

\begin{table} [h]
     \centering
     \subtable[ARL for one-sided moving averages]{
          \label{tbl:ma:s1}
                \begin{tabular}{|c||cc|cc|cc|}
                        \hline
                $\delta\rightarrow$ &\multicolumn{2}{|c|}{2}&\multicolumn{2}{|c|}{2.5}&\multicolumn{2}{|c|}{3}\\
                          \hline
                $k\downarrow$  &  $L$    &   $L_{\lceil k/2 \rceil}$&  $L$    &   $L_{\lceil k/2 \rceil}$&  $L$    &   $L_{\lceil k/2 \rceil}$ \\
                        \hline \hline
                3   & 63.0   & 62.5   & 206.4  & 204.5  & 869.6  & 866.8  \\
                4   & 73.6   & 71.0   & 233.3  & 227.7  & 967.0  & 947.4  \\
                5   & 84.2   & 84.0   & 263.3  & 261.4  & 1055.8 & 1057.6 \\
                6   & 94.8   & 93.2   & 292.1  & 286.8  & 1155.8 & 1147.5 \\
                8   & 115.7  & 114.7  & 346.7  & 344.2  & 1353.0 & 1345.2 \\
                10  & 136.5  & 135.6  & 403.4  & 401.0  & 1548.8 & 1547.3 \\
                13  & 167.0  & 166.9  & 487.1  & 484.1  & 1835.3 & 1832.8 \\
                16  & 196.7  & 196.9  & 568.5  & 567.0  & 2119.5 & 2110.5 \\
                        \hline
                \end{tabular}
                } \\
     \subtable[ARL for one-sided filtered derivatives]{
          \label{tbl:fd:s1}
                \begin{tabular}{|c||cc|cc|cc|}
                        \hline
                $\delta\rightarrow$ &\multicolumn{2}{|c|}{2}&\multicolumn{2}{|c|}{2.5}&\multicolumn{2}{|c|}{3}\\
                          \hline
                $k\downarrow$  &  $L$    &   $L_{\lceil k/2 \rceil}$&  $L$    &   $L_{\lceil k/2 \rceil}$&  $L$    &   $L_{\lceil k/2 \rceil}$ \\
                        \hline \hline
                4   & 47.7   & 49.3   & 166.4  & 168.4  & 749.3  & 752.1  \\
                6   & 54.3   & 56.5   & 181.0  & 183.7  & 788.3  & 791.9  \\
                8   & 61.7   & 64.5   & 198.7  & 202.1  & 842.0  & 846.8  \\
                10  & 69.5   & 72.6   & 217.5  & 221.6  & 902.0  & 908.0  \\
                12  & 77.1   & 80.7   & 237.3  & 241.6  & 968.0  & 972.8  \\
                14  & 84.8   & 88.9   & 256.6  & 261.7  & 1033.6 & 1036.9 \\
                16  & 92.6   & 97.0   & 276.4  & 281.8  & 1098.6 & 1106.1 \\
                        \hline
                \end{tabular}
                }
     \caption{Comparison of the ARL with its series approximation. For both moving average and filtered derivative algorithms of span $k$, the $\lceil k/2\rceil$\sups{th} order approximation is reasonably accurate.}
     \label{tbl:mafd:s1}
\end{table}

\section{Conclusion} \label{sec:conc}
In this paper, we have considered the approximation of ARL for moving sum algorithms with arbitrary weights using multivariate probabilities. Specifically, we have considered moving average and filtered derivative algorithms. We have applied a series approach that was originally proposed by \cite{Rob78} for geometric moving average algorithms. We have presented an analysis of the convergence of the series. We have shown using simulation studies that multivariate probabilities of order $\lceil k/2 \rceil$ can provide reasonable approximations of the ARL. We have also derived the ARL for two special cases of MOSUM algorithms, and have used them as illustrative examples.

\bibliographystyle{elsarticle-harv}
\bibliography{series}
\end{document}